\documentclass[letterpaper, 10 pt, conference]{ieeeconf}  

\IEEEoverridecommandlockouts                              
\usepackage{graphicx} 
\usepackage{subfig}
\usepackage{times} 
\usepackage{amsmath} 
\usepackage{amssymb}  
\usepackage{cite}
\usepackage{color}
\usepackage{url}

\newcommand{\KeyGen}{\mathsf{Gen}}
\newcommand{\Ecd}{\mathsf{Ecd}}
\newcommand{\Dcd}{\mathsf{Dcd}}
\newcommand{\Enc}{\mathsf{Enc}}
\newcommand{\Encv}{\mathsf{Enc}_v}
\newcommand{\EncG}{\mathsf{Enc}_G}
\newcommand{\Dec}{\mathsf{Dec}}

\newcommand{\sk}{\mathsf{sk}}
\newcommand{\Setup}{\mathsf{Setup}}

\newcommand{\BitDecomp}{G^{-1}}

\newcommand{\uniform}{\mathcal{U}(\mathbb{Z}_q)}
\newcommand{\discrete}{\mathcal{D}(0,\sigma)}
\newcommand{\UpdateKey}{\mathcal{T}_{\mathcal{K}}}
\newcommand{\Updatecv}{\mathcal{T}_{\mathcal{C}_v}}
\newcommand{\UpdateCG}{\mathcal{T}_{\mathcal{C}_G}}

\newtheorem{definition}{Definition}

\newtheorem{lemma}{Lemma}
\newtheorem{theorem}{Theorem}

\newtheorem{problem}{Problem}
\newtheorem{remark}{Remark}

\title{\LARGE \bf
Dynamic-Key Post-Quantum Encrypted Control  Against \\ System Identification Attacks
}

\author{Jungjin Park and Kiminao Kogiso
\thanks{This work was supported by JSPS KAKENHI Grant Numbers 22H01509 and 23K22779.}
\thanks{J. Park and K. Kogiso are with the Department of Mechanical and Intelligent Systems Engineering, Graduate School of
Informatics and Engineering, The University of Electro-Communications, Chofu, Tokyo,
Japan {\tt\small \{parkjungjin,kogiso\}@uec.ac.jp}}%
}
\begin{document}

\maketitle
\thispagestyle{empty}
\pagestyle{empty}

\thanks{This work has been submitted to the IEEE for possible publication.
Copyright may be transferred without notice, after which this version may no longer be accessible.}

\begin{abstract}
This study proposes post-quantum encrypted control systems based on dynamic-key Learning with Errors (LWE) encryption schemes. 
The proposed method develops update maps that simultaneously update the private key and ciphertexts within the LWE framework, enabling dynamic-key encrypted control resistant to system identification attacks. 
The growth of errors induced by homomorphic operations is analyzed, and sufficient parameter conditions guaranteeing correct decryption at each control step are clarified. 
Furthermore, a design procedure for the encrypted control systems is presented based on security metrics such as sample-identifying complexity and deciphering time. 
A numerical example demonstrates that the proposed control systems achieve secure control against the considered system identification attack.
\end{abstract} 
\section{Introduction}
Cyber-physical systems (CPSs) are considered core technologies for realizing Industry~5.0 and expected to support a wide range of applications, such as digital-twin-based power grid optimization and smart city infrastructure~\cite{Industry5.0survey22}.
However, CPSs are exposed to various cyber threats, including eavesdropping on control and measurements, signal tampering, and replay attacks~\cite{Teixerira2015,CPSsecuritiy2019}.
These vulnerabilities have motivated the development of secure control frameworks that protect both control signals and system information.

Encrypted control has attracted attention as an effective approach for mitigating these threats~\cite{Kogiso15,Junsoo2016,Darup2018}.
In encrypted control, homomorphic encryption is used to protect controller parameters and signals transmitted over communication networks while still allowing control computations to be performed on encrypted data.
However, the rapid progress of quantum computing raises concerns about the long-term security of conventional public-key cryptosystems, such as RSA and ElGamal~\cite{Gerjuoy04,Proos03}.
To address this issue, cryptosystems based on the Learning with Errors (LWE) problem have been proposed and are considered resistant to quantum attacks~\cite{Regev05,Peikert09}.
Several post-quantum cryptosystems (PQC), including the Regev encryption~\cite{Regev05} and the GSW~\cite{GSW}, possess homomorphic properties and have therefore been investigated for encrypted control applications~\cite{Junsoo23,Robert24}.
Nevertheless, even when PQC-based encryption schemes are employed, adversaries may still collect encrypted input-output data and attempt to infer system information.

Protecting the system dynamics is therefore crucial for ensuring the security of control systems.
If an adversary obtains accurate knowledge of the system dynamics, sophisticated attacks exploiting the system model can be designed.
In particular, stealthy attacks, such as covert attacks and zero-dynamics attacks, are known to destabilize control systems while remaining difficult to detect~\cite{Pasqualetti2013,Chong2019}.
To obtain such knowledge, adversaries may perform a system identification attack, estimating the system dynamics from collected input–output data.
To evaluate the security of encrypted control systems against such attacks, recent studies have introduced security notions such as sample-identifying complexity and deciphering time~\cite{Teranishi2023,Teranishi23least}.
The former represents the number of samples required for an adversary to identify the system with a given accuracy, while the latter represents the computation time required to decrypt the collected samples.
Based on these metrics, system designers can determine security parameters by considering the expected operational lifetime of the system and the allowable estimation accuracy.
However, when static-key encryption schemes are employed, adversaries can accumulate encrypted input–output data over time and attempt to recover system information from the collected samples.

One promising approach to enhancing security against system identification attacks is to use a dynamic-key encryption scheme, in which the private key and ciphertexts are updated at each control step.
In this framework, each transmitted sample is encrypted with a different key, making it significantly more difficult for adversaries to accumulate usable data for system identification.
Furthermore, dynamic-key encryption can mitigate attacks based on previously recorded data, such as signal falsification attacks and replay attacks~\cite{kosugi24}.
However, existing dynamic-key encrypted control methods are limited to the ElGamal encryption scheme~\cite{Teranishi20}, and dynamic-key encryption schemes based on PQC have not yet been investigated.

The objective of this study is to propose a secure encrypted state-feedback control system based on a dynamic-key LWE encryption scheme.
The key idea is to develop updated maps that simultaneously update both the private key and the ciphertexts within the LWE encryption framework.
Under the proposed scheme, we analyze the growth of errors induced by homomorphic operations and derive sufficient conditions on parameters that guarantee correct decryption at each control step.
Moreover, we present a design procedure for the secure encrypted control system based on sample-identifying complexity and deciphering time.
A numerical example demonstrates that the proposed controller achieves secure control against the considered system identification attack.

The contributions of this study are summarized as follows:
(i) We construct secure encrypted state-feedback control systems that prevent system identification attacks by incorporating dynamic-key encryption into the control loop;
(ii) We develop a dynamic-key LWE-based encryption scheme by introducing transition maps that jointly update the private key and ciphertexts, thereby extending dynamic-key encrypted control from the ElGamal scheme to a PQC framework;
(iii) We derive parameter conditions that suppress error growth caused by homomorphic operations and present a design procedure for secure encrypted control systems based on the sample-identifying complexity and the deciphering time.

The remainder of this paper is as follows.
Section~\ref{Sec:Pro} describes the problem formulation.
Section~\ref{Sec:DynLWE} develops the dynamic LWE-based encryption scheme.
Section~\ref{Sec:ecs} proposes the encrypted control systems.
Section~\ref{Sec:exp} provides a numerical example.
Finally, Section~\ref{Sec:con} concludes this paper.

\vspace*{1ex}
\noindent\textbf{Notation:} The sets of real numbers, rational numbers, integers, natural numbers, plaintext space, and ciphertext space are denoted by $\mathbb{R},\mathbb{Q},\mathbb{Z},\mathbb{N}, \mathcal{M},\mathcal{C}$, respectively.
The rounding and floor functions are denoted by $\left\lceil\cdot\right\rfloor$ and $\left\lfloor \cdot \right\rfloor$, respectively.
We define $\mathbb{Z}^+ := \{z \in \mathbb{Z} \; | \; 0 \leq z\}$, and $\mathbb{Z}_q := \{z \in \mathbb{Z} \; | \; -\tfrac{q}{2} < z \leq \tfrac{q}{2}\}$.
For a scalar $a \in \mathbb{R}$, its absolute value is denoted by $|a|$.
The set of vectors of size $n$ is denoted by $\mathbb{R}^n$. 
The $j$th element of a vector $v$ is denoted by $v_j$.
The infinity norm of $v$ is denoted by $\|v\|_\infty$.
The set of matrices of size $m \times n$ is denoted by $\mathbb{R}^{m \times n}$.
The $(i,j)$ entry of matrix $M$ is denoted by $M_{ij}$.
The Frobenius norm of $M\in\mathbb{R}^{m \times n}$ is denoted by $\|M\|_F:=\sqrt{\mathrm{tr}(M^\top M)}$.
For $r\in\mathbb{N}$, $I_r\in\mathbb{R}^{r\times r}$ denotes the identity matrix.
The uniform distribution over $\mathbb{Z}_q$ is denoted by $\uniform$, and the zero-mean discrete Gaussian distribution with standard deviation $\sigma>0$ is denoted by $\discrete$.
The gadget matrix $G$ is defined as
$G:=\begin{bmatrix} I_r & \nu I_r & \cdots & \nu^{d-1} I_r\end{bmatrix}$, where $\nu$ is the radix, and $d\in \mathbb{N}$ is chosen such that $\nu^{d-1}<q\le\nu^d$.
For an arbitrary vector $w=[w_0,\ldots,w_{r-1}]^{\top}\in\mathbb{Z}_q^{r}$, the digit decomposition map $\BitDecomp(\cdot)$ expands each entry of $w$ into its base-$\nu$ representation and is defined as $\BitDecomp(w):=[w_{0,0},\ldots, w_{0,d-1},\,w_{1,0}, \ldots,w_{r-1,d-1}]^{\top}\in\mathbb{Z}^{dr}$,
where $w_i=\sum_{j=0}^{d-1} w_{i,j}\nu^j$ with $w_{i,j}\in\{-\nu+1,\ldots,\nu-1\}$. 
With these definitions, the relation $w=G\,\BitDecomp(w) $ holds.

\section{Attack Scenario and Problem Formulation}\label{Sec:Pro}
This section describes the considered attack scenario and formulates the problem addressed in this study.

\subsection{Control System Setup}
We consider a discrete-time linear plant and a state-feedback control law as follows:
\begin{subequations}\label{eq:system}
\begin{align}
x(k+1) &=A_p x(k)+B_pu(k)+w(k), \\
u(k)   &=F x(k), 
\end{align}
\end{subequations}
where $k\in\mathbb{Z}^+$, $x\in\mathbb{R}^{n_p}$, $u\in\mathbb{R}$ and $w\in\mathbb{R}^{n_p}$ denote the step, state, control input, and process noise, respectively. 
Suppose $x(0)$ and $w(k)$ are are i.i.d. Gaussian random variables with mean $\boldsymbol{0}$ and covariance $\sigma_p^2 I_{n_p}$.
The pair $(A_p,B_p)$ is controllable, and $F\in\mathbb{R}^{n_p}$ is a feedback gain designed such that $A_p+B_pF$ is stable.

To convert real-valued signals into plaintext elements, we use the following encoder and decoder:
$\Ecd_{\Delta}: \mathbb{R}\ni x\mapsto\check{x}=\left\lceil\Delta x\right\rfloor\in\mathcal{M}$, and 
$\Dcd_{\Delta}: \mathcal{M}\ni\check{x}\mapsto \bar{x}=\frac{\check{x}}{\Delta}\in\mathbb{Q}$, respectively, where the quantization gain $\Delta > 0$ is chosen as a power of two. 
The quantization gain $\Delta$ should be selected such that the quantized values remain within the plaintext space, thereby preventing overflow or underflow.

\subsection{Attack Scenario}
We consider an attack scenario in which an adversary performs system identification using the least-squares method, as in~\cite{Teranishi23least}.
Suppose that the adversaries use a quantum computer to recover the plaintext input and state data by breaking the underlying PQC scheme.

The considered attack scenario is as follows: Against \eqref{eq:system} an adversary performs the following steps:
        1) The adversary eavesdrops the encrypted data generated by control system~\eqref{eq:system} from $k_s$ to $k_s+N$, where $k_s\in\mathbb{Z}^+$ and $N>0$ is the number of samples;
        2) The adversary decrypts the encrypted data to obtain the plaintext data $\{x(k_s),\cdots,x(k_s+N)\}$;
        3) The adversary estimates the dynamics $\mathsf{A}:=A_p+B_pF$ by the least squares method; $\hat{\mathsf{A}} = \mathrm{argmin}_{\mathsf{A} \in \mathbb{R}^{n_p \times n_p}} \left\| X_f - \mathsf{A} X_p \right\|_F^2 = X_f X_p^\dagger$, where $X_f = [x(k_s + 1) \ \cdots \ x(k_s + N )]$, $X_p = [x(k_s) \ \cdots \ x(k_s + N-1)]$, and $X_p^\dagger$ is the pseudoinverse of $X_p$. 

In addition, an estimation error is defined as
$\epsilon(N,F)=\tfrac{1}{n_p^2} \left\| \mathsf{A}-\hat{\mathsf{A}}\right\|_F^2$.
From~\cite{Teranishi23least}, the expected estimation error $\mathbb{E}[\epsilon(N,F)]$ can be characterized using the sample-identifying complexity of~\eqref{eq:system}. 
In particular, the minimum number of samples required to achieve an acceptable estimation accuracy can be derived from this complexity measure.
Suppose that the defender who designs the encrypted control system specifies an acceptable estimation accuracy $\gamma_c>0$. 
Then, the minimum sample size required for the adversary to achieve $\gamma_c$ is given by
$N^*(\gamma_c) = \left\lfloor \tfrac{n_p}{\gamma_c \, \mathrm{tr}(\Psi)} \right\rfloor + 2$, where $\Psi$ is the solution to the discrete Lyapunov equation $\mathsf{A}\Psi \mathsf{A}^\top - \Psi + I = 0$.
For simplicity, we assume that the quantization error $\bar{x}-x$ is negligible because of the appropriate choice of a quantization gain, and the sample size $N^*$ is publicly known.
This represents a worst-case scenario for the defender, in which identifying the closed-loop dynamics requires recovering $N^*$ plaintext samples.

\subsection{Security Concept of Encrypted Control Systems} %
To discuss the security of encrypted control systems, we focus on the computational effort required for an adversary to break encrypted data.
The security level of an encryption scheme, quantified in terms of bit security~\cite {Cryptography05}, is introduced.
\begin{definition}
    An encryption scheme satisfies $\lambda$-bit security if at least $2^{\lambda}$ operations are required to break the scheme.
\end{definition}
The security of encrypted control systems is defined by the computational time required to break ciphertexts used in system identification. 
This computation time is referred to as the sample-deciphering time in~\cite{Teranishi23least}, where quantum computers are not taken into account.
We define a quantum sample-deciphering time $T_Q$ using the quantum bit-security $\lambda_Q$ and processing capability of a quantum computer in Circuit Layer Operations Per Second (CLOPS)~\cite{CLOPS25}\footnote{The performance of currently available quantum computers is $3.40\times 10^{5}$ CLOPS, which is much slower than that of state-of-the-art classical computers, $1.81\times 10^{18}$ FLOPS. 
However, in this study, we adopt a conservative threat model by assuming an adversary equipped with a quantum computer equivalent in performance to current high-performance classical computers.}.
Let the adversary be equipped with a quantum computer capable of $\Xi$~CLOPS. 
The quantum sample-deciphering time is
\begin{align}\label{def:qunatum_decipher}
    T_Q(N, \lambda_Q) = \frac{N  2^{\lambda_Q}}{\Xi},
\end{align}
where $\lambda_Q$ is quantum bit security.
The quantum sample-deciphering time $T_Q$ represents the computational time required for the adversary to decrypt $N$ ciphertexts generated by a dynamic-key homomorphic encryption scheme that ensures $\lambda_Q$-bit quantum security for system identification.  
We use $T_Q$ to define the security of encrypted control systems as follows.

\begin{definition}\label{def:secure}
Let $T_c>0$ denote the operational lifetime of the control system~\eqref{eq:system}.
The control system designer (defender) determines the acceptable estimation accuracy $\gamma_c$ and the minimum sample size $N^*(\gamma_c)$.
If $T_Q(N^*(\gamma_c),\lambda_Q)>T_c$, then the encrypted control system is said to be secure under the attack model.
\end{definition}

\subsection{Problem Formulation}
In the considered attack scenario, the defender should design the encrypted control system so that the adversary cannot obtain sufficient data to perform system identification.
Accordingly, the design problem of secure encrypted control systems reduces to selecting security parameters and a key-update strategy such that $T_Q(N^*(\gamma_c),\lambda_Q) > T_c$, which can be formulated as follows.
\begin{problem}
Given the control system described in~\eqref{eq:system},
design an encrypted state-feedback control system that is secure in the sense of \textit{Definition~\ref{def:secure}}.
\end{problem}
We attempt to solve the problem by appropriately determining the update mechanism and by clarifying the conditions under which the encrypted and unencrypted control systems coincide.

\section{Dynamic-Key PQC}\label{Sec:DynLWE}
This section develops a novel dynamic-key LWE encryption scheme and confirms its properties, including correctness and full homomorphism.

\subsection{Dynamic-key LWE-based Encryption Scheme}
The key-updatable LWE-based encryption scheme is presented as follows.
A dynamic-key LWE-based  encryption scheme $\Pi_{\mathrm{dyn}}$ at time step $k$ is defined as the tuple,
\begin{align*}
\Pi_{\mathrm{dyn}}
:= (\Setup, \KeyGen, \Enc_v, \Enc_G, \Dec, \UpdateKey, \Updatecv, \UpdateCG),
\end{align*}
where $(\Setup, \KeyGen, \Enc_v, \Enc_G, \Dec)$ are the algorithms of a conventional LWE-based  encryption scheme~\cite{Junsoo20},
and the remaining tuple ($\UpdateKey$, $\Updatecv$, $\UpdateCG$) are update transition maps.
The algorithms are described as follows.
 \begin{itemize}
    \item $\Setup(1^\lambda)$: Choose a dimension $n\in\mathbb{N}$,
    modulus $t=2^{t_0}$, $t_0\in\mathbb{N}$, modulus $q=2^{q_0}$, $q_0\in\mathbb{N}$, such that $t\ll q$, and a standard deviation $\sigma>0$.
    Choose $d\in\mathbb{N}$ and define $\nu:=2^{\nu_0}$ such that $\nu^{d-1}<q\le \nu^d$, and set $n_G:=d(n+1)$. Define the plaintext space $\mathcal{M}:=\mathbb{Z}_t$, and the ciphertext spaces $\mathcal{C}_v:=\mathbb{Z}_q^{n+1}$, $\mathcal{C}_G:=\mathbb{Z}_q^{(n+1)\times n_G}$. 
    Return $\mathsf{p}=(n,t,q,\sigma,\nu,d,n_G)$.

    \item $\KeyGen(\mathsf{p})$: Generate the private key $\sk\in\mathbb{Z}_q^n$ as a column vector, with each element sampled from $\uniform$.
    Return $\sk$ and $\tau = \begin{bmatrix}1\\ \sk\end{bmatrix}$.

    \item $\Encv(m\in\mathbb{Z}_t,\sk)$: Given input $m\in\mathbb{Z}_t$ and $\sk$, return the ciphertext $c\in\mathbb{Z}_q^{n+1}$, 
    where $a\in\mathbb{Z}_q^n$ is a random column vector and $e\in\mathbb{Z}$ is an error sampled from $\discrete$,  $b=-\sk^{\top}a+\tfrac{q}{t}m+e \bmod q$, and $c=\begin{bmatrix}b\\ a\end{bmatrix}$.
    For simplicity, the ciphertext of a vector $X$ encrypted component-wise by $\Encv$ is denoted by $c_X$.
    
    \item $\EncG(m\in\mathbb{Z}_t,\sk)$: Given input $m\in\mathbb{Z}_t$ and $\sk$, return the ciphertext $C\in\mathbb{Z}_q^{(n+1)\times n_G}$,
    where $A\in\mathbb{Z}_q^{n\times n_G}$ is a random matrix, and $E\in\mathbb{Z}_q^{n_G}$ is an error row vector, with each element sampled from $\discrete$, $B=-\sk^{\top}A+E \bmod q$, and $C=mG+\begin{bmatrix}B\\ A\end{bmatrix}\bmod q$.
    For simplicity, the ciphertext of a vector $X$ encrypted component-wise by $\EncG$ is denoted by $C_X$.

    \item $\Dec(c\in\mathbb{Z}_q^{n+1},\sk)$: Given ciphertext $c$ and private key $\sk$, compute $\left\lceil\tfrac{t}{q}\tau^\top c\right\rfloor = \left\lceil m + \tfrac{t}{q}e \right\rfloor $. If $\tfrac{t}{q}|e|<\tfrac{1}{2}$, then $\Dec(c,\sk)=m$.

    \item $\UpdateKey$: The private key update map is defined as
    \begin{align*}
      &\UpdateKey:\ (\sk(k), \tau(k)) \mapsto (\sk(k+1), \tau(k+1)) \\
      &\quad = \left(\sk(k)+s(k)\bmod q, \begin{bmatrix}1\\ \sk(k)+s(k)\end{bmatrix}\bmod q\right),
    \end{align*}
    where $s\in\mathbb{Z}_q^n$ is a random column vector with each element sampled from $\uniform$.

    \item $\Updatecv$: The ciphertext update map for $c$ is defined as
    \begin{align*}
      \Updatecv:\ c(k) \mapsto c(k+1)
      =\begin{bmatrix}
        b-s(k)^{\top}a\\
        a
      \end{bmatrix}\bmod q,
    \end{align*}
    where $c(k)$ is the ciphertext encrypted by $\Encv$ and $\sk(k)$, and $s(k)\in\mathbb{Z}_q^n$ is a random column vector with each element sampled from $\uniform$ at step $k$.

    \item $\UpdateCG$: The ciphertext update map for $C$ is defined as
    \begin{align*}
      \UpdateCG: C(k)\mapsto &C(k+1)  \\
      &  =mG+\begin{bmatrix}
        B-s(k)^{\top}A\\
        A
      \end{bmatrix}\bmod q,
    \end{align*}
    where $C(k)$ is the ciphertext encrypted by $\EncG$ and $\sk(k)$, and $s(k)\in\mathbb{Z}_q^n$ is a random column vector with each element sampled from $\uniform$ at step $k$.
\end{itemize}

\begin{remark}
In contrast to~\cite{Junsoo20}, the developed framework introduces update transition maps, namely $\UpdateKey$, $\Updatecv$, and $\UpdateCG$, which enable dynamic-key operation in encrypted control systems. 
These update mechanisms allow the encryption keys and related cryptographic states to be updated during control operations, thereby enhancing the security of the control systems against long-term cryptanalytic or data-driven attacks.
\end{remark}

In LWE-based cryptosystems, the applicable attack methods and the resulting quantum-bit security depend on the parameters $(n,q,\sigma)$ and the number of ciphertexts generated under a given private key, denoted by $\mathsf{m}$. 
This dependence makes it difficult to assess the security based on a single attack model. 
Therefore, we evaluate the quantum-bit security $\lambda_Q$ using the Lattice Estimator (LE)~\cite{LWE_estimator}, which computes the bit security levels for all applicable attacks. 
For example, applying the LE to the parameters $(\mathsf{m},n,q,\sigma) = (239,239,2^{48},1.0)$ used in~\cite{Junsoo23} yields a minimum quantum-bit security of $2^{41.1}$, and thus $\lambda_Q = 41$ bits.

\begin{remark}
In the encrypted control system based on the proposed dynamic-key LWE encryption scheme, the number of ciphertexts accessible to an adversary under a single private key is inherently limited due to the key-update mechanism.
This implies that the parameter $\mathsf{m}$ is effectively bounded in practical operations.
Therefore, security evaluation and parameter selection can be further specialized to encrypted control systems by accounting for this structural restriction on $\mathsf{m}$.
A systematic design of security parameters is left for future work.
\end{remark}

\subsection{Properties of Dynamic-key LWE-based Encryption} 
We clarify that the developed dynamic-key LWE-based encryption scheme satisfies correctness and fully homomorphic properties at each step.

\begin{lemma}\label{prop:correctness}
Let $\sk(k)$ denote the private key at step $k$ and $c(k) = \Encv(m,\sk(k))$.
For $m\in\mathbb{Z}_t$, if the parameters $t$ and $q$ satisfy $|\tfrac{t}{q} e(k)| < \tfrac{1}{2}$, then $ \Dec(c(k),\sk(k)) = m$, where $e(k)$ is the error generated by $\Encv$.
\end{lemma}

\begin{proof}
Let $c(k) = \begin{bmatrix} b \\ a \end{bmatrix} \in \mathbb{Z}_q^{n+1}$ with
$b = -\sk(k)^\top a + \frac{q}{t} m + e(k) \bmod q$.
Then, $\Dec(c(k), \sk(k))=\lceil \frac{t}{q} \tau(k)^\top c(k) \rfloor =\lceil \frac{t}{q} \bigl( b + \sk(k)^\top a \bigr) \rfloor =\lceil m + \frac{t}{q} e(k) \rfloor$.
Since $|\tfrac{t}{q}e(k)| < \tfrac{1}{2}$, the decryption yields $m$.
\end{proof}

Next, the homomorphic operations are defined as follows.
\begin{definition}
For ciphertexts $c_1, c_2 \in \mathbb{Z}_q^{n+1}$, define the homomorphic addition as $c_1 \oplus c_2 := c_1 + c_2$.
Moreover, for ciphertext $C \in \mathbb{Z}_q^{(n+1)\times n_G}$ and $c \in \mathbb{Z}_q^{n+1}$, define the homomorphic multiplication as $C \otimes c := C \BitDecomp(c)$.
\end{definition}

The following lemma shows that the homomorphic addition and multiplication are correctly supported under the dynamic-key LWE-based encryption.

\begin{lemma}\label{prop:dyn_add}
Let $\sk(k)$ be the private key at step $k$, and let $c_1(k) = \Encv(m_1,\sk(k))$ and $c_2(k) = \Encv(m_2,\sk(k))$ be ciphertexts with errors $e_1(k)$ and $e_2(k)$ generated by $\Encv$, respectively.  
For $m_1, m_2 \in \mathbb{Z}_t$, if $t$ and $q$ satisfy $|\tfrac{t}{q} e_{\mathsf{add}}(k)| < \tfrac{1}{2}$, then $\Dec(c_1(k) \oplus c_2(k), \sk(k)) = m_1 + m_2$, where $e_{\mathsf{add}}(k) := e_1(k) + e_2(k)$.
\end{lemma}

\begin{proof}  
By \textit{Lemma~\ref{prop:correctness}}, the decryption of their homomorphic sum is given by
$\Dec(c_1(k) \oplus c_2(k), \sk(k)) =\lceil \frac{t}{q} \, \tau(k)^\top \bigl(c_1(k) + c_2(k)\bigr) \rfloor = \lceil m_1 + m_2 + \frac{t}{q} e_{\mathsf{add}}(k) \rfloor$,
where $e_{\mathsf{add}}(k) := e_1(k) + e_2(k)$.
Since $|\tfrac{t}{q} e_{\mathsf{add}}(k)| < \tfrac{1}{2}$, thus$ \lceil m_1 + m_2 + \frac{t}{q} e_{\mathsf{add}}(k) \rfloor = m_1 + m_2$.
\end{proof}

\begin{lemma}\label{prop:dyn_mult}
Let $\sk(k)$ be the private key at step $k$, $C(k) = \EncG(m_1, \sk(k))$, and $c(k) = \Encv(m_2, \sk(k))$, with errors $E(k)$ and $e(k)$ generated by $\EncG$ and $\Encv$, respectively.  
For $m_1, m_2 \in \mathbb{Z}_t$, if $t$ and $q$ satisfy $|\tfrac{t}{q} e_{\mathsf{mult}}(k)| < \tfrac{1}{2}$, then $\Dec(C(k) \otimes c(k), \sk(k)) = m_1 m_2$, where $e_{\mathsf{mult}}(k) := m_1 e (k) + E(k) \BitDecomp(c(k))$.
\end{lemma}

\begin{proof}
Let $C(k)$ be a ciphertext encrypted by $\EncG(m_1, \sk(k))$, and $c(k)$ by $\Encv(m_2, \sk(k))$, with corresponding errors $E(k)$ and $e(k)$.  
By $\EncG$,
\begin{align*}
  C(k) = m_1 G + \begin{bmatrix} -\sk(k)^\top A + E(k)\\ A \end{bmatrix} \bmod q.    
\end{align*}
Then, the decryption of the homomorphic product yields 
$\Dec(C(k) \otimes c(k), \sk(k)) = \lceil \frac{t}{q} \, \tau(k)^\top C(k) \BitDecomp(c(k))\rfloor=\lceil \tfrac{t}{q}m_1\tau(k)^\top c(k) + \tfrac{t}{q} E(k)\BitDecomp(c(k)) \rfloor$, since $G\BitDecomp(c(k)) = c(k)$.
Moreover, by \textit{Lemma~\ref{prop:correctness}}, it follows that $ \lceil \tfrac{t}{q}m_1 \tau(k)^\top c(k) + \tfrac{t}{q} E(k) \BitDecomp(c(k)) \rfloor = \lceil m_1 m_2 + \tfrac{t}{q} \bigl( m_1 e(k) + E(k) \BitDecomp(c(k)) \bigr) \rfloor = \lceil m_1m_2+\tfrac{t}{q}e_{\mathsf{mult}}\rfloor$.
If $|\tfrac{t}{q} e_{\mathsf{mult}}(k)| < \tfrac{1}{2}$, then $\lceil m_1m_2 + \tfrac{t}{q} e_{\mathsf{mult}}(k) \rfloor = m_1m_2$.
\end{proof}

\section{Design of Dynamic-Key Encrypted State-Feedback Control Systems}\label{Sec:ecs}
This section proposes encrypted state-feedback control systems using the proposed dynamic-key LWE-based encryption scheme and analyzes the errors introduced by homomorphic operations. 
This analysis presents parameter conditions and a design procedure for secure encrypted control systems.

\subsection{Encrypted State-Feedback Controller} 
Following the LWE-based encryption scheme framework introduced in \cite{Junsoo20}, we construct encrypted state-feedback control systems based on a dynamic-key LWE-based encryption scheme as follows.
 Let $C_F(k)$, $c_x(k)$, and $c_u(k)$ denote the ciphertexts of the gain $F$, the state $x(k)$, and the control input $u(k)$ at time step $k \in \mathbb{Z}^+$, respectively.
 Consider the control system~\eqref{eq:system} with the controller~$f$.
 Then, an encrypted controller, $f_{\Pi_{\mathsf{dyn}}}:\ (C_F(k),c_x(k))
 \mapsto c_u(k)$, is given by
 \begin{align}
  \sk(k+1)&=\UpdateKey(\sk(k)),\quad  C_F(k+1)=\UpdateCG(C_F(k)),\notag\\
 c_u(k)&=\bigoplus_{j=1}^{n_p}C_{F_{j}}(k)\otimes c_{x_j}(k)\notag,
 \end{align}
 where $\mathsf{p}=\Setup(1^{\lambda})$, $\sk(0)=\KeyGen(\mathsf{p})$, $C_F(0)=\EncG(\Ecd_{\Delta}(F),\sk(0))$, $c_x(k)=\Encv(\Ecd_{\Delta}(x(k)),\sk(k))$, 
 and for a vector, $\Encv$, $\EncG$, $\Dec$, $\Ecd_{\Delta}$, and $\Dcd_{\Delta}$ perform element-wise operations.
 The encrypted controller performs computations on the encrypted data $F$, $x(k)$, and $u(k)$ using a dynamic private key $\sk(k)$.

The above construction follows the proposed dynamic-key LWE-based encryption scheme and enables encrypted state-feedback control under a time-varying private key.
For practical implementation, the ciphertext update procedure $\UpdateCG$ requires secure access to both the random matrix $A$ and vector $s(k)$.

\subsection{Error Analysis and Design Procedure}
This section analyzes the error introduced by homomorphic operations in the encrypted control and derives parameter conditions under which the encrypted control system coincides with the quantized (or original) control system. 
The design procedure is then described based on the analysis.

Correct decryption requires that the total error $e_{\mathsf{total}}$ satisfies the condition $\tfrac{t}{q}|e_{\mathsf{total}}| < \tfrac{1}{2}$.
If this condition is violated, the decrypted control input may differ from the plaintext value.
Therefore, the parameters $t$ and $q$ must be designed to accommodate the total error introduced by homomorphic operations and ensure correct decryption.

To analyze the worst-case growth of the error, we introduce a bound on the error magnitude.
Let $\kappa \in \mathbb{N}$ be a sufficiently large constant, and neglect the tail probability that $|e| > \kappa\sigma$ for each error $e$ sampled from the discrete Gaussian distribution $\discrete$.
Thus, we impose the bound $|e| \leq \kappa\sigma$.
Applying these error bounds to the proposed encrypted controller, the following theorem provides a sufficient condition for correct decryption.
\begin{lemma}\label{thm:designt1}
If, for some $\kappa$, the parameters $\Delta$, $t$, and $q$ satisfy
\begin{align}\label{ieq:design_t1}
    n_p \bigl(\|\Ecd_{\Delta}(F)\|_\infty + n_{G}\nu \bigr)\kappa \sigma < \frac{q}{2t}, 
\end{align}
then encrypted control systems eliminate the accumulated errors introduced by homomorphic operations at each step.
\end{lemma}

\begin{proof}
From \textit{Lemmas~\ref{prop:dyn_add}} and \ref{prop:dyn_mult}, the total error in the encrypted controller is given by $e_{\mathsf{total}} = \sum_{j=1}^{n_p} e_{j_{\mathsf{mult}}}$, where $e_{j_{\mathsf{mult}}}$ denotes the error introduced by the $j$-th homomorphic multiplication.
From \textit{Lemma~\ref{prop:dyn_mult}}, each homomorphic multiplication introduces an error $e_{j_\mathsf{mult}} = \Ecd_{\Delta}(F_j) e + E \BitDecomp(c_{x_j}(k))$.
The first term is bounded as $|\Ecd_{\Delta}(F_j)e| \le \|\Ecd_{\Delta}(F)\|_\infty |e|\le \|\Ecd_{\Delta}(F)\|_\infty\kappa\sigma,$ since $|e|<\kappa\sigma$.
Moreover, the second term can be bounded as $|E\BitDecomp(c_{x_j}(k))|\le n_G\nu \kappa\sigma $, since $|E_i|\le \kappa\sigma$ and $\BitDecomp(c_{x_j}(k))$ has $n_G$ elements.
Combining these bounds, the multiplication error $|e_{\mathsf{mult}}|$ satisfies $|e_{\mathsf{mult}}|\le (\|\Ecd_{\Delta}(F)\|_\infty+n_G\nu)\kappa\sigma$.
Summing over the $n_p$ multiplications yields the accumulated error, $e_{\mathsf{total}} = \sum_{j=1}^{n_p} e_{j_{\mathsf{mult}}} \leq n_p (\|\Ecd_{\Delta}(F)\|_\infty + n_G\nu) \kappa \sigma$.
From \textit{Lemma~\ref{prop:correctness}}, correct decryption requires $\frac{t}{q} |e_{\mathsf{total}}|<\frac{1}{2}$.
Substituting the above bound yields $\frac{t}{q}|e_{\mathsf{total}}|\leq\frac{t}{q}n_p (\|\Ecd_{\Delta}(F)\|_\infty+n_G\nu)\kappa\sigma <\frac{1}{2}$, which is equivalent to $n_p(\|\Ecd_{\Delta}(F)\|_\infty+n_G\nu)\kappa \sigma<\frac{q}{2 t}$. 
This yields the inequality~\eqref{ieq:design_t1}.
\end{proof}
\textit{Lemma~\ref{thm:designt1}} implies that the total error introduced by homomorphic operations can be completely eliminated, thereby ensuring correct decryption.
Based on this result, the following theorem establishes the equivalence between the encrypted and quantized control systems, as illustrated in Fig.~\ref{fig1}.

\begin{theorem}\label{cor1}
 If the parameters $\Delta$, $t$, and $q$ satisfy~\eqref{ieq:design_t1}, then the encrypted control systems $f_{\Pi_{\mathsf{dyn}}}$ are equivalent to the quantized control systems with the quantized state-feedback gain $\bar{F}$ and state $\bar{x}(k)$. 
 Moreover, under the assumption that quantization errors are negligible, this implies that the encrypted control systems $f_{\Pi_{\mathsf{dyn}}}$ are equivalent to the original systems~\eqref{eq:system}.
\end{theorem}

\begin{proof}
Let $\bar{u}(k)=\Dcd_{\Delta}(\Dec(c_u(k),\sk(k)))$ denote the decrypted control input.
From \textit{Lemma~\ref{thm:designt1}}, correct decryption is ensured, and thus the above expression becomes
$\bar{u}(k)=  \Dcd_{\Delta}\!\left(\sum_{j=1}^{n_p}\Ecd_{\Delta}(F_{j})  \Ecd_{\Delta}(x_j(k))\right) = \bar{F}\bar{x}(k)$.
Therefore, the decrypted control input is equivalent to the output of the quantized state-feedback controller.
Moreover, under the assumption, it holds that $\bar{F}=F$ and $\bar{x}(k)=x(k)$, which results in $\bar{u}(k)=u(k), \forall k\in\mathbb{Z}^+$.
\end{proof}

Consequently, \textit{Lemma~\ref{thm:designt1}} and \textit{Theorem~\ref{cor1}} provide the design procedure for state-feedback control systems~\eqref{eq:system} as follows: 
i) Set $\gamma_c$, $T_c$, $\kappa$, and $\Xi$;
ii) Compute $N^*(\gamma_c) = \left\lfloor \tfrac{n_p}{\gamma_c \, \mathrm{tr}(\Psi)} \right\rfloor+2$ and $\lambda_Q = \left\lfloor\log_2\left(\tfrac{\Xi\,T_c}{N^*(\gamma_c)} \right)\right\rfloor + 1$ from~\eqref{def:qunatum_decipher};
iii) Search the parameters $(q,\sigma,n,\nu,d,n_G)$ according to $\lambda_Q$ using the LE;
iv) Choose $(\Delta, t)$ such that inequality~\eqref{ieq:design_t1} is satisfied.

\begin{figure}[tb]
      \centering
  \subfloat[Encrypted control system]{
\includegraphics[width=0.45\linewidth,keepaspectratio]{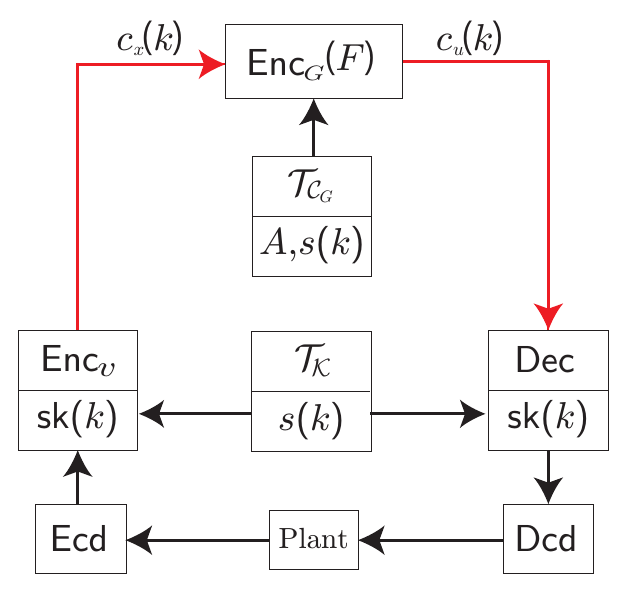}
  }
    \hfill
  \subfloat[Quantized control system]{
\includegraphics[width=0.43\linewidth,keepaspectratio]{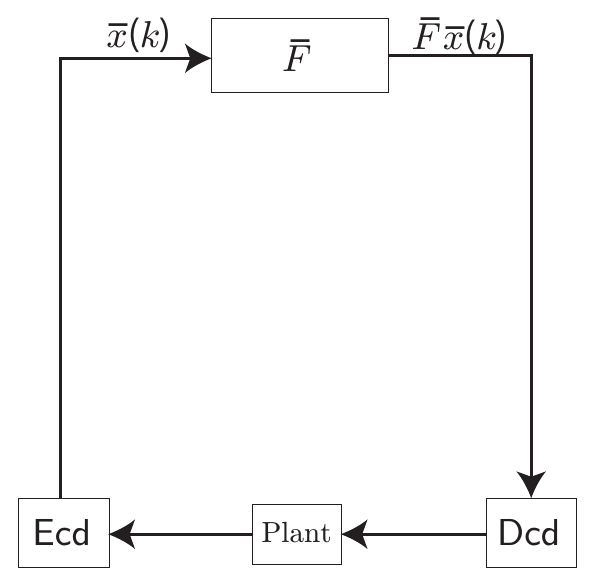}
}
 \caption{Configurations of encrypted and quantized control systems.}
 \label{fig1}
 \end{figure}

\section{Numerical Example}\label{Sec:exp}
This section demonstrates the design of a secure encrypted control system through numerical simulations.
The computations were performed on a MacBook Air equipped with an Apple M4 chip and 24\,GB of memory, running macOS Tahoe 26.2.
The encryption scheme was implemented using the C++ NTL library.

Consider the system~\eqref{eq:system} with
\begin{align*}
    A_p = \begin{bmatrix} 1 & 0.5 \\ 0 & -1.2 \end{bmatrix},\  
    B_p = \begin{bmatrix} 0 \\ 1 \end{bmatrix}, \ 
    F = \begin{bmatrix}-0.5198 & 0.6158 \end{bmatrix},
\end{align*}
where $\sigma_p = 0.01$.
Following the design procedure, 
(i) the parameters $\gamma_c$, $T_c$, $\kappa$, and $\Xi$ were set to $10^{-5}$, $ 3.154 \times 10^{8}$~s (10 years), $6$, and $1.81 \times 10^{18}$~CLOPS\footnote{This value corresponds to the performance of the El Capitan supercomputer https://www.top500.org/lists/top500/2025/11/}, respectively;
ii) this results in $N^* = 33198$ and $\lambda_Q = 74$ bits;
iii) using~\cite{LWE_estimator},  $(n,q,\sigma,\nu,d,n_G) = (5,2^{50},2.0,2,50,300)$, where the number of ciphertexts was set to $\mathsf{m}=5$; and
iv) the resulting parameters $(\Delta,t)$ were $(2^{16},2^{29})$.
As a result, the quantum sample-deciphering time satisfies $T_Q(33198,74) = 3.465 \times 10^8 > T_c$, which indicates that the designed control system is secure against the system identification attack.

The output of the encrypted controller is compared with that of the quantized controller to validate \textit{Theorem~\ref{cor1}}.
Fig.~\ref{fig:result} shows the difference between the outputs of the encrypted and quantized controllers. 
Fig.~\ref{fig:result}(a) corresponds to the parameters obtained by the presented procedure, while Fig.~\ref{fig:result}(b) corresponds to the parameters $(\Delta,t)=(2^{26},2^{50})$, which do not satisfy \textit{Lemma~\ref{thm:designt1}}.
These results indicate that the parameters $(\gamma,t)$ designed using the proposed procedure yield outputs that coincide with those of the quantized controller, which supports \textit{Theorem~\ref{cor1}}. 
In contrast, when $t=q$, a difference on the order of $10^{-7}$ to $10^{-8}$ is observed, indicating that the error introduced by the homomorphic operations remains.
This example confirms that encrypted control systems designed using the presented procedure coincide with the corresponding quantized control systems, in which no LWE-induced errors occur during operations.

\begin{figure}[tb]
      \centering
  \subfloat[$(\Delta,t) = (2^{16},2^{29})$]{
\includegraphics[width=0.47\linewidth,keepaspectratio]{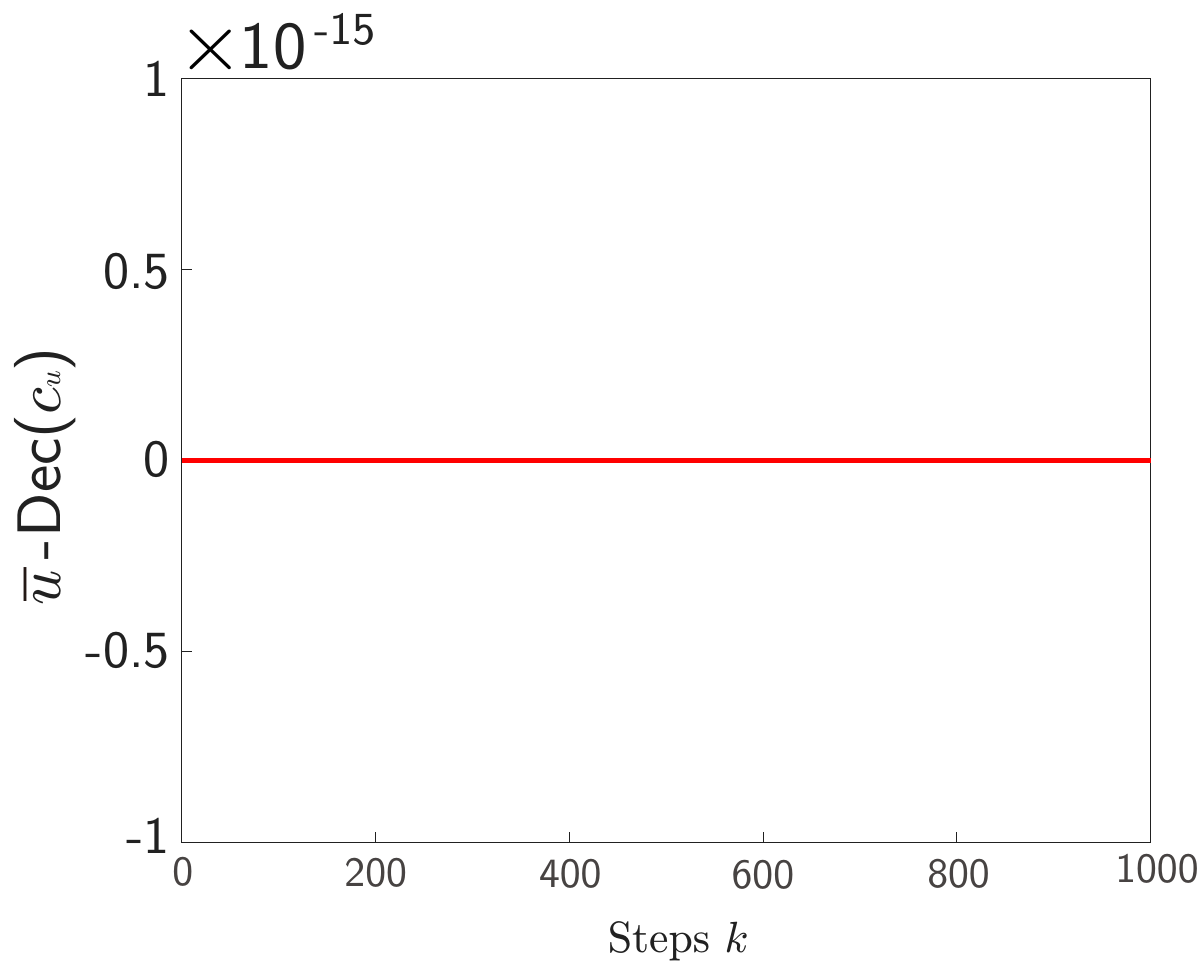}
  }
    \hfill
  \subfloat[$(\Delta,t) = (2^{26},2^{50})$]{
\includegraphics[width=0.47\linewidth,keepaspectratio]{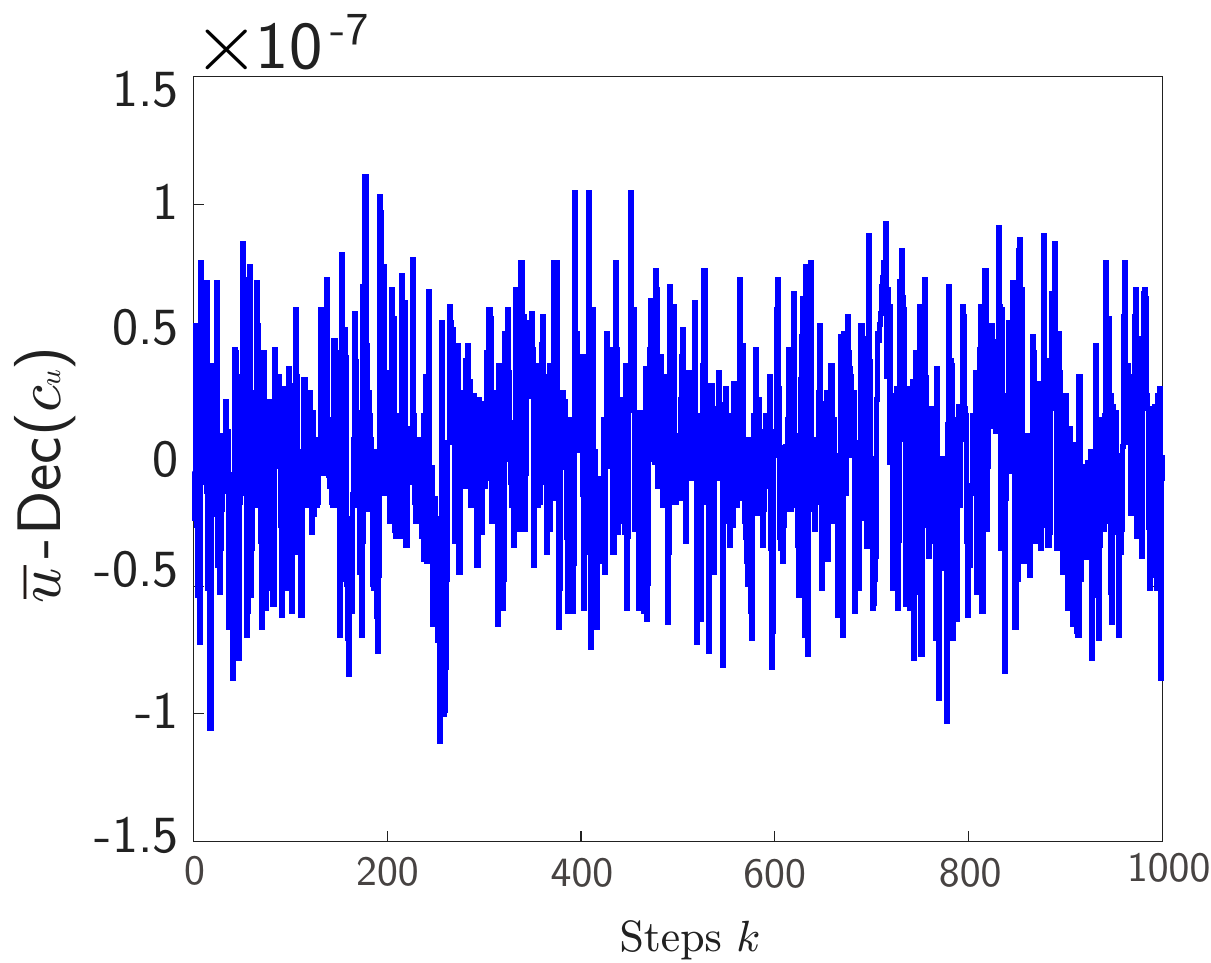}
  }
 \caption{The difference between the encrypted controller's output and that of the quantized controller.}
 \label{fig:result}
 \end{figure}

\section{Conclusion}\label{Sec:con}
This study proposed post-quantum encrypted control systems based on a dynamic-key Learning with Errors (LWE) encryption scheme.
We developed update maps that simultaneously update the private key and ciphertexts within LWE encryption and constructed encrypted state-feedback control systems based on the proposed scheme.
The growth of errors induced by homomorphic operations was analyzed, and sufficient conditions on the parameters that guarantee correct decryption at each step were derived.
Based on this analysis, we presented a design procedure for secure encrypted control systems using security metrics such as sample-identifying complexity and deciphering time.
A numerical example demonstrated that the secure encrypted controller coincides with the corresponding quantized control system, with no errors induced by the LWE-based encryption.

Future work includes extending the proposed framework to general discrete-time linear controllers beyond state-feedback control.
In addition, applying the dynamic-key concept to ring-LWE-based post-quantum cryptographic schemes such as BFV and CKKS is an important direction for further research.

\section*{ACKNOWLEDGMENTS}
The authors are grateful to Dr. Kenji Yasunaga, Associate Professor, Department of Mathematical and Computing Science, School of Computing, Institute of Science Tokyo, for his helpful and valuable comments.

\bibliographystyle{IEEEtran}
\bibliography{references}
\end{document}